\documentclass[11pt]{article}
\usepackage[a4paper,margin=1in]{geometry}
\usepackage{amsmath,amssymb,amsfonts,amsthm}
\usepackage{graphicx}
\usepackage{microtype}
\usepackage{booktabs}
\usepackage{enumitem}
\usepackage{mathtools}
\usepackage{csquotes}
\usepackage[authoryear,round]{natbib}
\usepackage[authoryear,round]{natbib}
\usepackage[hidelinks]{hyperref}

\title{A Timeless Game: A Game-Theoretic Model of Mass–Geometry Relations}
\author{Milad Ghadimi \\
        Technische Universit\"at Dresden \\
        01069 Dresden, Germany \\
        \texttt{mil.ghadimi@gmail.com}}

\theoremstyle{plain}
\newtheorem{proposition}{Proposition}
\newtheorem{definition}{Definition}

\newcommand{\Oplayer}{\ensuremath{\mathsf{O}}}
\newcommand{\Splayer}{\ensuremath{\mathsf{S}}}
\begin{document}
\maketitle

\begin{abstract}
We develop a minimal, \emph{timeless} game–theoretic representation of the mass–geometry relation. An “Object” (mass) and “Space” (geometry) choose strategies in a static normal–form game; utilities encode \emph{stability} as mutual consistency rather than dynamical payoffs. In a $2\times 2$ toy model the equilibria correspond to \enquote{light–flat} and \enquote{heavy–curved} configurations; a continuous variant clarifies when only trivial interior equilibria appear versus a continuum along a matching ray. Philosophically, the point is representational: a global description may be static while the \emph{experience} of temporal flow for embedded observers arises from informational asymmetry, coarse–graining, and records. The framework separates \emph{time as parameter} from \emph{relational constraint} without committing to specific physical dynamics.

\end{abstract}
\section{Introduction}
A long-standing question in the philosophy of physics is whether time is fundamental or merely an emergent bookkeeping device for relational change. From McTaggart’s argument for the unreality of time \citep{mctaggart1908unreality}, through Barbour’s notion of a timeless physics \cite{barbour1999end}, to block-universe interpretations of relativity \citep{einstein1916,rovelli2018order}, a long tradition suggests that a \emph{global} description of the universe may be static, while \emph{local} observers experience a \enquote{flow} due to their embeddedness and limited perspective.

In this spirit, we present a compact formalization of that idea using static game theory. We model an \enquote{Object} (mass distribution) and \enquote{Space} (geometry) as players in a normal-form (one-shot) game, where payoffs encode relational \emph{stability} or self-consistency. No explicit temporal parameter appears. Nonetheless, the Nash equilibria of the game correspond to \emph{frozen} configurations analogous to stationary solutions in physics.

\paragraph{Contributions.}
\begin{enumerate}[leftmargin=*,itemsep=2pt,topsep=2pt]
\item We introduce \emph{Curvature Clash}, a two-player timeless game with physically motivated strategies for mass and curvature.
\item We prove that a simple $2\times 2$ instantiation has two pure-strategy Nash equilibria, corresponding to \enquote{flat–light} (low-mass flat-space) and \enquote{curved–heavy} (high-mass curved-space) configurations.
\item We analyze a continuous variant, identify best responses, and clarify when the only equilibrium is trivial versus when a continuum of equilibria appears.
\item We discuss how embedded observers can recover an apparent temporal ordering from static equilibria via asymmetries and information constraints. We analyze a continuous variant, derive best responses, and identify conditions under which the interior equilibrium is trivial or forms a nontrivial continuum.
\end{enumerate}
\paragraph{Scope.} The model is a deliberately austere toy. Payoffs are heuristic stability scores, not derived from an action principle; no quantum constraints are imposed; and the two–agent idealization abstracts from locality and anisotropy. The value lies in clarifying a thesis about representation (timeless global description) and about epistemology (emergent temporal phenomenology).

\section{Background}
\subsection{Philosophy of Time}
Throughout history, philosophers have debated the nature of time. Parmenides argued that reality is one, unchanging, and timeless, with change being illusory \citep{parmenides_stanford}. In contrast, Heraclitus emphasized flux, viewing change as fundamental \citep{heraclitus_stanford}. St. Augustine posited that time was created with the universe, with God existing outside of time \citep{augustine_time}.

In the modern era, Newton advocated for absolute time, an independent entity, while Leibniz viewed time as relational, derived from events \citep{newton_leibniz_stanford}. Mach echoed this, stating that \enquote{time is nothing but change} \citep{mach_principle}. McTaggart distinguished the A-series (events as past, present, future) and B-series (earlier-later relations), arguing that time leads to contradictions and is unreal \citep{mctaggart1908unreality}.

Contemporary views include presentism (only the present exists) \citep{presentism_stanford}, the growing-block theory (past and present exist, future does not) \citep{broad_growing_block}, and eternalism (all times exist equally in a block universe) \citep{eternalism_stanford}.
\subsection{Timelessness in Modern Physics}
Minkowski’s 1909 unification makes the tenseless picture explicit: spacetime is a four-dimensional geometric whole with no objective global \enquote{now} \citep{minkowski1909space}. Einstein's relativity fused space and time into spacetime, supporting the block universe where past, present, and future are illusions \citep{einstein1916,putnam_relativity}. Einstein remarked that the distinction between past, present, and future is a \enquote{stubbornly persistent illusion} \citep{einstein_quote}. 

In quantum gravity, the Wheeler-DeWitt equation lacks a time parameter, leading to the problem of time \citep{wheeler_dewitt}. Barbour's timeless universe posits configurations in \enquote{Platonia}, with time as an illusion \citep{barbour1999end}. Rovelli's loop quantum gravity views time as emergent from quantum relations, via the thermal time hypothesis \citep{rovelli_thermal_time}. The Page-Wootters mechanism derives time from entanglement \citep{page_wootters}.
Dissenters like Smolin argue time is fundamental \citep{smolin2013time}.

\section{The Model: \emph{Curvature Clash}}
\begin{definition}[\emph{Curvature Clash}]
Let the player set be $I = \{\Oplayer,\Splayer\}$, with strategy sets
$S_{\Oplayer}=\{L,H\}$ (low/high mass), $S_{\Splayer}=\{F,C\}$ (flat/curved geometry).
Payoffs for each player represent \emph{stability scores}, i.e., how self-consistent the chosen mass--curvature combination is as a global configuration (higher is better).
\end{definition}
A concrete specification is given by the payoff matrix in Table~\ref{tab:payoff}, where each cell lists the payoff pair $(u_{\Oplayer}, u_{\Splayer})$.
\begin{table}[h!]
\centering
\begin{tabular}{@{}lcc@{}}
\toprule
& $F$ (Flat) & $C$ (Curved) \\
\midrule
$L$ (Low mass) & $(5,5)$ & $(3,2)$ \\
$H$ (High mass) & $(2,3)$ & $(4,4)$ \\
\bottomrule
\end{tabular}
\caption{A minimal payoff specification for \emph{Curvature Clash}. }
\label{tab:payoff}
\end{table}

\begin{proposition}[Equilibria of the $2\times 2$ model]
The game in Table~\ref{tab:payoff} has two pure-strategy Nash equilibria: $(L, F)$ and $(H, C)$.
\end{proposition}
\begin{proof}
Considering best responses: If $\Splayer$ plays $F$, then $\Oplayer$ prefers $L$ (since $5 > 2$). If $\Splayer$ plays $C$, then $\Oplayer$ prefers $H$ (since $4 > 3$). Conversely, if $\Oplayer$ plays $L$, $\Splayer$ prefers $F$ (since $5 > 2$); and if $\Oplayer$ plays $H$, $\Splayer$ prefers $C$ (since $4 > 3$). Thus, $(L, F)$ and $(H, C)$ are mutual best responses and hence Nash equilibria.
\end{proof}
\paragraph{Interpretation.} The equilibrium $(L, F)$ corresponds to a self-consistent \enquote{light-mass flat-space} configuration (an almost empty, flat universe), whereas $(H, C)$ represents a \enquote{heavy-mass curved-space} configuration (a strongly gravitating regime). Both are \emph{timeless} in the sense that nothing \enquote{evolves}: we are identifying only static, stable relations between mass and geometry.
\subsection{A Continuous Variant}
To explore a graded strategy space, let
$s_{\Oplayer}\in[0,m]$, $s_{\Splayer}\in[0,k]$
denote the choice of mass level and curvature degree, respectively. Consider payoff functions of the form
\begin{align}
u_{\Oplayer}(s_{\Oplayer}, s_{\Splayer}) &= -\Big(s_{\Oplayer} - \alpha s_{\Splayer}\Big)^2 + \beta, \label{eq:uO}\\
u_{\Splayer}(s_{\Oplayer}, s_{\Splayer}) &= -\Big(s_{\Splayer} - \gamma s_{\Oplayer}\Big)^2 + \delta, \label{eq:uS}
\end{align}
where $\alpha,\gamma > 0$ encode the expected \enquote{matching} between mass and curvature, and $\beta,\delta$ are positive constants (baseline payoff levels).

\begin{proposition}[Best responses and equilibria in the continuous model]
For interior choices (ignoring boundary constraints), the best-response functions are
$\mathrm{BR}_{\Oplayer}(s_{\Splayer})=\alpha s_{\Splayer}$, $\mathrm{BR}_{\Splayer}(s_{\Oplayer})=\gamma s_{\Oplayer}$.
The simultaneous Nash equilibria satisfy $s_{\Oplayer} = \alpha s_{\Splayer}$ and $s_{\Splayer} = \gamma s_{\Oplayer}$, hence $(1 - \alpha\gamma) s_{\Oplayer} = 0$. Consequently:
\begin{enumerate}[label=(\roman*), leftmargin=*]
\item If $\alpha \gamma \neq 1$, the unique equilibrium (in the interior of $[0,m]\times[0,k]$) is the trivial solution $s_{\Oplayer}^* = s_{\Splayer}^* = 0$.
\item If $\alpha \gamma = 1$, there exists a continuum of equilibria along the ray $s_{\Oplayer}^* = \alpha s_{\Splayer}^*$ (for values within the allowed range $[0,m]\times[0,k]$).
\end{enumerate}
\end{proposition}
\begin{proof}
From \eqref{eq:uO}, holding $s_{\Splayer}$ fixed, $u_{\Oplayer}$ is a concave quadratic in $s_{\Oplayer}$ maximized when $s_{\Oplayer} = \alpha s_{\Splayer}$. Similarly, \eqref{eq:uS} is maximized (for fixed $s_{\Oplayer}$) when $s_{\Splayer} = \gamma s_{\Oplayer}$. Setting these best responses equal gives $s_{\Oplayer} = \alpha s_{\Splayer}$ and $s_{\Splayer} = \gamma s_{\Oplayer}$. This system has the solutions stated: if $\alpha\gamma \neq 1$, the only solution is $s_{\Oplayer}=s_{\Splayer}=0$; if $\alpha\gamma = 1$, any $(s_{\Oplayer}, s_{\Splayer})$ satisfying $s_{\Oplayer} = \alpha s_{\Splayer}$ is a solution (subject to bounds).
\end{proof}
\paragraph{Remarks.} The quadratic coupling form \eqref{eq:uO}--\eqref{eq:uS} is intentionally minimal and tends to select the trivial equilibrium unless the fine-tuned coupling condition $\alpha\gamma=1$ holds. To generically obtain a nontrivial interior equilibrium, one can incorporate additional linear benefits and self-costs. For example, consider modified payoffs:
\begin{equation*}
u_{\Oplayer} = -a\big(s_{\Oplayer} - \alpha s_{\Splayer}\big)^2 + b s_{\Oplayer} - c s_{\Oplayer}^2, \qquad
u_{\Splayer} = -d\big(s_{\Splayer} - \gamma s_{\Oplayer}\big)^2 + e s_{\Splayer} - f s_{\Splayer}^2,
\end{equation*}
with $a,d,c,f > 0$ and $b,e \ge 0$. These payoffs yield affine best-response functions (linear in the other player’s strategy), which under mild conditions will intersect at a unique interior point $(s_{\Oplayer}^*, s_{\Splayer}^*) \in (0,m)\times(0,k)$. In short, slight generalizations of the utility functions can ensure a unique nontrivial equilibrium without requiring $\alpha\gamma = 1$ exactly.
When $\alpha\gamma\neq 1$, our minimal quadratic “matching” form selects the trivial interior equilibrium $(0,0)$. Depending on bounds and added linear terms, corner equilibria may arise. In the affine extension just given, the linear best–response curves generically intersect at a unique interior point when the product of their slopes is less than one.
\section{Discussion: Emergent Temporal Flow}
Our framework is \emph{timeless}: equilibria encode relational consistency conditions rather than dynamical trajectories. Yet observers embedded in subsystems and working with coarse descriptions can experience an apparent \enquote{before/after}. Three generic mechanisms suffice:
\begin{enumerate}[leftmargin=*,itemsep=2pt,topsep=2pt]
    \item \textbf{Asymmetric information.} Limited access to the global configuration forces sequential inference, which is \emph{experienced} as evolution.
    \item \textbf{Coarse-graining and entropy.} Ordering macrostates by increasing entropy induces an arrow of time at the descriptive level.
    \item \textbf{Records and memory.} Local records pick out consistent slices of a static whole, allowing observers to reconstruct histories.
\end{enumerate}
This picture coheres with block-universe intuitions in relativity \citep{einstein1916,rovelli2018order} and with timeless, relational accounts in philosophy of physics \citep{barbour1999end}. It also resonates (purely conceptually) with emergent-time themes in quantum gravity without presupposing any such theory \citep{page_wootters,wheeler_dewitt,rovelli_thermal_time}.

\section{Limitations and Scope}
\textbf{Abstraction.} Payoffs are heuristic stability scores rather than derived from an action principle; they serve to illustrate a thesis about representation, not to predict phenomena.\\
\textbf{No quantum dynamics.} We draw conceptual parallels to emergent-time programs, but the model does not include quantum constraints (e.g., Wheeler–DeWitt).\\
\textbf{Two-agent coarse idealization.} Matter and geometry are field-like and distributed; our two-player reduction ignores locality and anisotropy.\\
\textbf{Equilibrium selection.} Multiple equilibria (discrete case) and trivial interior equilibria (minimal continuous case) show that selection is a separate, philosophical issue about rational constraint—not addressed here.

\section{Conclusion}
\emph{Curvature Clash} recasts the mass–geometry link as a problem of timeless consistency. The discrete model yields two natural equilibria (\enquote{light–flat}, \enquote{heavy–curved}); the continuous matching form either collapses to a trivial interior point or, under exact coupling, generates a continuum along a matching ray. The upshot is conceptual: one can cleanly separate (i) \emph{global description as static} from (ii) \emph{temporal phenomenology as epistemic and coarse–grained}. Time, as an external parameter, is not required to articulate self–consistent mass–geometry relations, even though temporal experience remains intelligible for embedded observers.

\begingroup
\bibliographystyle{plainnat} 
\bibliography{sn-bibliography}

\end{document}